\documentclass[letterpaper, 10 pt, conference]{ieeeconf}  

\IEEEoverridecommandlockouts                              
\usepackage{amsmath} 
\usepackage{amssymb}  
\usepackage{subcaption}
\usepackage{multirow}
\usepackage{bm}
\usepackage{todonotes}
\newtheorem{theorem}{Theorem}

\newtheorem{proposition}{Proposition}
\newtheorem{definition}{Definition}
\newtheorem{lemma}{Lemma}
\newtheorem{corollary}{Corollary}

\DeclareMathOperator{\rank}{rank}
\DeclareMathOperator{\nsc}{nsc}             
\DeclareMathOperator{\supp}{supp}           
\DeclareMathOperator{\img}{im}              
\renewcommand{\subset}{\subseteq}
\newcommand{\range}[1]{[\![#1]\!]}          
\newcommand{\RR}{\mathbb{R}}

\newcommand{\OO}{\bm{\mathcal{O}}}       
\renewcommand{\AA}{\bm{A}}          
\newcommand{\CC}{\bm{C}}            
\newcommand{\II}{\bm{I}}            
\newcommand{\PP}{\bm{P}}            
\newcommand{\PPsi}{\bm{\Psi}}       
\newcommand{\TTheta}{\bm{\Theta}}   
\newcommand{\GGamma}{\bm{\Gamma}}   

\newcommand{\yy}{\bm{y}}            
\newcommand{\xx}{\bm{x}}            
\newcommand{\uu}{\bm{u}}            
\newcommand{\vv}{\bm{v}}            
\newcommand{\hh}{\bm{h}}            

\newcommand{\UU}{\mathcal{U}}
\newcommand{\VV}{\mathcal{V}}

\title{\LARGE \bf
Necessary and Sufficient Conditions for Simultaneous State and Input Recovery of Linear Systems with Sparse Inputs by $\ell_1$-Minimization
}

\author{Kyle Poe, Enrique Mallada, Ren\'e Vidal
\thanks{K. Poe is with the Department of Biomedical Engineering, Johns Hopkins University,
MD 21218
        {\tt\small kpoe2@jh.edu}}%
\thanks{E. Mallada is with the Department of Electrical and Computer Engineering, Johns Hopkins University,
MD 21218
        {\tt\small mallada@jhu.edu}}%
        \thanks{R. Vidal is with the Department of Electrical and Systems Engineering, University of Pennsylvania,
PA 19104
        {\tt\small vidalr@seas.upenn.edu}}%
    }

\begin{document}
\bstctlcite{MyBSTcontrol} 

\maketitle
\thispagestyle{empty}
\pagestyle{empty}

\begin{abstract}

    The study of theoretical conditions for recovering sparse signals from compressive measurements has received a lot of attention in the research community. In parallel, there has been a great amount of work characterizing conditions for the recovery both the state and the input to a linear dynamical system (LDS), including a handful of results on recovering sparse inputs. However, existing sufficient conditions for recovering sparse inputs to an LDS are conservative and hard to interpret, while necessary and sufficient conditions have not yet appeared in the literature. In this work, we provide (1) the first characterization of necessary and sufficient conditions for the existence and uniqueness of sparse inputs to an LDS, (2) the first necessary and sufficient conditions for a linear program to recover both an unknown initial state and a sparse input, and (3) simple, interpretable recovery conditions in terms of the LDS parameters. We conclude with a numerical validation of these claims and discuss implications and future directions.
\end{abstract}

\section{INTRODUCTION}

A foundational concept in systems theory is that of \textit{observability}, the condition guaranteeing uniqueness of the initial state of a system given knowledge of the inputs and a sufficient number of measurements for the output \cite{kalmanMathematicalDescriptionLinear1963}. Introduced later was the more stringent notion of \textit{strong observability}, which further guarantees the uniqueness of the initial condition even in the presence of unknown inputs, and is known to be equivalent to the system having no invariant zeros \cite{payneDiscreteTimeAlgebraic1973a}. These conditions have been used to concisely characterize conditions under which either the initial state or inputs to a system, or both, can be recovered, even in the absence of the other \cite{houInputObservabilityInput1998, ansariDeadbeatUnknowninputState2019, kurekStateVectorReconstruction1983a, molinariStrongControllabilityObservability1976}. Of particular relevance to time-critical applications is the development of deadbeat or finite-time input reconstructors, which in the discrete setting have been formulated in terms of solutions to a block Toeplitz system \cite{ansariDeadbeatUnknowninputState2019}. 

A linear system is, in particular, a compact means of describing a linear relationship $\yy = \PPsi\uu$ between a sequence of of inputs $\uu$ and a sequence of observations $\yy$, from which even in the most optimistic circumstances generic $\uu$ can only be reconstructed up to $\ker\PPsi$.
However, by assuming that $\uu$ is sparse in an appropriate sense, established results in sparse recovery provide favorable guarantees on exact reconstruction for an appropriate choice of optimization algorithm.
The most common such case considered is when $\uu$ is assumed to have support of size not greater than $s$, and is termed \textit{regular sparsity}. Other support-based notions of sparsity include block \cite{eldarBlockSparsityCoherenceEfficient2008}, group \cite{jenattonStructuredVariableSelection}, and tree-based sparsity \cite{heExploitingStructureWaveletBased2009}, and are each subsumed by the more general notion of model-based sparsity \cite{baraniukModelBasedCompressiveSensing2010}.

For each of these sparsity patterns, the literature has provided tailored optimization problems and recovery guarantees, with varying levels of robustness to noise, ease of checking, and conceptual nuance. For many applications in the noiseless setting, simple $\ell_1$-minimization has proven to be the approach of choice due to its relative conceptual ease, implementability as a linear program, and favorable performance even when compared with tailored regularizers \cite{elhamifarBlockSparseRecoveryConvex2012}. For regular sparsity, the necessary and sufficient condition for successful unique recovery is the satisfaction of the so-called \textit{nullspace property} (NUP), which requires vectors in the nullspace of $\PPsi$ to have smaller $\ell_1$ norm on $s$-sparse supports than on the complement of $s$-sparse supports (see Def.~\ref{def:NUP}). Recent results have even shown that for any support-based notion of sparsity, there exists a straightforward extension of the NUP, termed the \textit{generalized nullspace property}, which provides necessary and sufficient guarantees \cite{kabaWhatLargestSparsity2020}.

In light of the success of this approach to signal reconstruction, recent literature has provided tailored algorithms for sparse recovery in linear dynamical systems (LDSs), where the assumption of sparsity has been variously made on the initial conditions \cite{wakinObservabilityLinearSystems2010, josephObservabilityLinearSystem2018, josephMeasurementBoundsObservability2019, daiObservabilityLinearSystem2013}, dynamics \cite{charlesDynamicFilteringTimeVarying2016}, measurement noise \cite{kahlLocalizationInvariableSparse2021a}, and inputs \cite{asifEstimationDynamicUpdating2011, sefatiLinearSystemsSparse2015, kafashanRelatingObservabilityCompressed2016,  josephNoniterativeOnlineBayesian2017}. Even with all of this prior work, there are few existing guarantees on the performance of these algorithms, and the guarantees that have been produced are typically probabilistic in nature or make restrictive assumptions on the sparsity patterns, such as the state and inputs being simultaneously sparse with respect to an orthogonal dictionary. As a result, many results for the general, noiseless setting, including the establishment of necessary and sufficient guarantees, have not yet appeared in the literature. Our focus in this work is on establishing such guarantees for the basis-pursuit style optimization problem introduced in \cite{sefatiLinearSystemsSparse2015}, where the initial state is not sparse, but the inputs are assumed to follow an appropriate generalized support pattern. Existing conditions for even the basic version of this problem are very conservative, and necessary and sufficient conditions have not yet made an appearance.

In this work, we consider the problem of jointly inferring the initial state $\xx_0$ and sparse inputs $U_N = (\uu_0, \ldots, \uu_{N-1})$ of an LDS 
$\Sigma = (\AA, \PPsi, \CC)$ without a feedthrough term, i.e., 
\begin{equation}
\label{eq:system}
    \begin{aligned}
        \xx_{k+1} &= \AA\xx_k + \PPsi \uu_k, & \xx_k \in \mathbb{R}^n, & \uu_k \in \mathbb{R}^m\\
        \yy_k &= \CC\xx_k, & \yy_k \in \mathbb{R}^p,
    \end{aligned}
\end{equation}
from $N+1$ output measurements $Y_N = (\yy_0, \ldots, \yy_N)$. 
In particular, this work makes the following contributions:
\begin{enumerate}
    \item Necessary and sufficient conditions for uniqueness of $\xx_0$ and sparse $U_N$ given $Y_N$.
    \item Necessary and sufficient conditions for the $\ell_1$-minimization approach of \cite{sefatiLinearSystemsSparse2015} to uniquely recover $\xx_0$ and sparse $U_N$ given $Y_N$.
    \item Interpretable conditions which are respectively necessary or sufficient for unique solutions to the $\ell_1$-minimization approach of \cite{sefatiLinearSystemsSparse2015} in terms of the system parameters, and elaboration on situations where these conditions achieve equality.
    \item Illustration of
    the accuracy of these conditions and provision of intuition for when they are most informative through simulations of random LDSs.
\end{enumerate}

\section{PRELIMINARIES}

\subsection{Notation}
\subsubsection{Sets and Vector Spaces}
Define $\mathbb{N} = \{0, 1, 2, \ldots \}$, $\range{n}:= \{0, 1, 2, \ldots, n-1\}$. We use capital script letters to denote vector subspaces $\mathcal{V} \subset \mathbb{R}^n$.
When $\VV, \UU \subseteq \mathbb{R}^n$, we denote $\VV + \UU := \{\vv + \uu: \vv \in \VV, \uu \in \UU\} \subseteq \mathbb{R}^n$ and $\VV^\perp := \{\xx \in \mathbb{R}^n: \forall \vv \in \VV, \langle \xx, \vv\rangle = 0\}$. 

\subsubsection{Linear Operators}
For any matrix $\AA$ and subspace $\UU$, define $\AA\UU := \{\AA\xx: \xx \in \UU\}$.
$\AA^{-1}$ is defined to be the inverse matrix of $\AA$ if it exists, and for any affine subspace $\VV \subset \mathbb{R}^n$, $\AA^{-1} \VV = \{ \xx \in \mathbb{R}^m: \AA \xx \in \VV\}$. We likewise denote the Moore-Penrose pseudoinverse as $\AA^+$.

\subsubsection{Supports and Norms}
For $\xx \in \mathbb{R}^m$, denote $\supp(\xx) := \{i \in \range{m}: x_i \ne 0\}$. Denote column $i$ of a matrix $\AA$ to be $\AA_i$, and block column $i$ if $\AA$ is a block matrix. For any subset $S \subset \range{m}$, $\AA_S$ is the submatrix of $\AA \in \mathbb{R}^{n \times m}$ with columns $(\AA_S)_i = \AA_{S_i}$. If $S = (S_k)$ is a tuple of sets and $\GGamma$ is a block matrix, denote $\GGamma_S$ the block matrix with block columns $(\GGamma_S)_k = (\GGamma_k)_{S_k}$. For a block vector $U$, $(U_S)_k = (U_k)_{S_k}$. Likewise if $S, S'$ are two tuples of sets with the same length, define $S\cup S'$ the tuple of sets s.t. $(S\cup S')_k = S_k \cup S_k'$. We denote the $\ell_0$ semi-norm as $\|\xx\|_0 = |\supp(\xx)|$ and the $\ell_1$ and $\ell_2$ norms as $\|\cdot\|_1$ and $\|\cdot\|_2$, 
respectively.

\subsection{The Nullspace Property}
\noindent Given a matrix $\TTheta \in \mathbb{R}^{n \times m}$, where we assume that $\TTheta$ has linearly dependent columns ($\rank{\TTheta} < m$),
a central problem of sparse recovery 
is to inquire, under which assumptions on the support of $\xx$ are there unique solutions to $\yy = \TTheta \xx$, and what are the algorithms with such unique recovery guarantees? A standard approach is to begin with the optimization problem $P_0$ that finds the sparsest solution, known to be NP-hard, and proceed to the convex relaxation $P_1$, known as basis pursuit:
\begin{align*}
    \min_{\xx} \|\xx\|_0,\mbox{ such that } \yy = \TTheta \xx \tag{$P_0$} \\
    \min_{\xx} \|\xx\|_1,\mbox{ such that } \yy = \TTheta \xx \tag{$P_1$}
\end{align*}

Denote $\Delta_s(m) := \{S \subset \range{m}: |S| \le s\}$ to be the set of $s$-sparse supports for vectors in $\mathbb{R}^m$, or simply $\Delta_s$ when $m$ is fixed. A classic result in sparse recovery is that any $s$-sparse solution to $P_1$ is the unique solution, if and only if $\TTheta$ satisfies the $s$-NUP:
\begin{definition}[Nullspace Property (NUP)]
\label{def:NUP}
    The matrix $\TTheta \in \mathbb{R}^{n \times m}$ satisfies the \textit{nullspace property of order $s$} ($s$-NUP) if $\forall \hh \in \ker{\TTheta} \setminus \{0\}$, $\forall S \in \Delta_s$, $\|\hh_S\|_1 < \|\hh_{S^c}\|_1$.
\end{definition}

\subsection{The Generalized Nullspace Property}
The motivating observation of \cite{kabaWhatLargestSparsity2020} is that sparsity structures tend to satisfy the property that if $S$ is a valid sparse support, so too is $S' \subset S$. This relationship describes an abstract simplicial complex:

\begin{definition}[Abstract Simplicial Complex (ASC)]
    Let $\Delta$ be a set of sets. $\Delta$ is an \emph{abstract simplicial complex} if $\forall S \in \Delta$, $\forall S' \subset S$, $S' \in \Delta$. If for some $m \in \mathbb{N}$, $\Delta \subset \{S: S \subset \range{m}\}$, we say that $\Delta$ is an ASC over $\range{m}$.
\end{definition}

One can quickly check that $\Delta_s(m)$ is an ASC over $\range{m}$, so ASCs comprise a strict generalization of regular sparsity. We will thus refer to any vector $\xx$ such that $\supp{\xx} \in \Delta$ as $\Delta$-sparse.
We additionally make the convenient definition:
\begin{equation}
    \mathcal{S}(\Delta) := \{\xx \in \RR^m: \supp \xx \in \Delta\}
    \end{equation}
which may be geometrically interpreted as the union of subspaces spanned by basis vectors $\{e_{k}\}_{k \in S}$, where $S \in \Delta$.
The associated result of \cite{kabaWhatLargestSparsity2020} is key:
\begin{definition}[Generalized Nullspace Property (GNUP)]
    Let $\PPsi \in \mathbb{R}^{n \times m}$, and let $\Delta$ be an ASC over $\range{m}$. We say that $\PPsi$ satisfies the \textit{generalized nullspace property with respect to $\Delta$} ($\Delta$-NUP) if $\forall \hh \in \ker{\PPsi} \setminus \{0\}$, $\forall S \in \Delta$, $\|\hh_S\|_1 < \|\hh_{S^c}\|_1$. Equivalently, 
    \begin{equation}
    \nsc(\PPsi, \Delta) := \max_{S \in \Delta} \max_{\hh \in \ker{\PPsi} \setminus \{0\}} \frac{\|\hh_S\|_1}{\|\hh\|_1} < \frac{1}{2}
    \end{equation}
    where $\nsc(\PPsi, \Delta)$ is called the \textit{nullspace constant}.
\end{definition}

\begin{proposition}\label{prop:gnup}
    Let $\Delta$ be a simplicial complex over $\range{m}$. Then any $\Delta$-sparse solution to $P_1$ is the unique solution, if and only if $\TTheta$ satisfies the $\Delta$-NUP.
\end{proposition}

This result potentially enables necessary and sufficient conditions for much more general types of sparsity patterns than are classically admissible for $P_1$. As a simple example, one can remark that the $\Delta_s$-NUP is equivalent to the $s$-NUP, so this is a generalization; but it also encompasses e.g. group sparsity. Since the GNUP is essentially a statement about the kernel of a particular matrix, it is natural to extend this characterization to any subspace with an appropriate choice of basis:
for a subspace $\UU \subset \mathbb{R}^n$ define $$\nsc(\UU, \Delta)  := \max_{S \in \Delta} \max_{\hh \in \UU \setminus \{0\}} \frac{\|\hh_S\|_1}{\|\hh\|_1}$$ and say that $\UU$ satisfies the $\Delta$-NUP if $\nsc(\UU, \Delta) < \frac{1}{2}$. 

\subsection{Linear Dynamical Systems with Sparse Inputs}

In this section, we will state the main problem of the paper. For an LDS $\Sigma = (\AA,\PPsi, \CC)$ whose state-space equations are defined in \eqref{eq:system}, we define the associated block matrices:
\begin{equation}
\begin{split}
    \OO_N &= \begin{bmatrix}
        \CC \\ \CC\AA \\ \vdots \\ \CC\AA^{N}
    \end{bmatrix},\, Y_N = \begin{bmatrix}
        \yy_0 \\ \yy_1 \\ \vdots \\ \yy_{N}
    \end{bmatrix}, \, U_N = \begin{bmatrix}
        \uu_0 \\ \uu_1 \\ \vdots \\ \uu_{N-1}
    \end{bmatrix} \\ \GGamma_N &= \begin{bmatrix}
        0 & 0 & \cdots & 0 \\
        \CC\PPsi & 0 & \cdots & 0 \\
        \CC\AA\PPsi & \CC\PPsi & \cdots & 0 \\
        \vdots & \vdots & \ddots\\
        \CC\AA^{N-1}\PPsi & \CC\AA^{N-2} & \cdots & \CC\PPsi
    \end{bmatrix},
    \end{split}
\end{equation}
where we refer to $\OO_N$ as the \textit{observability matrix} of $\Sigma$ and to $\GGamma_N$ as the \textit{input-output matrix}.
The reader should note the discrepancy in the number of entries of $Y_N$ and $U_N$; we have opted to admit this asymmetry for notational convenience. 

It follows that
\begin{equation}
    Y_N = \OO_N \xx_0 + \GGamma_N U_N, 
\end{equation}
 Of central interest in this paper is the case where $\forall k$, $\supp{\uu_k} \in \Delta$, where $\Delta$ is a simplicial complex over $\range{m}$. In this case, we say that the input $\uu$ and respectively the block vector $U_N$ is \textit{entrywise $\Delta$-sparse}; we may equivalently refer to $U_N$ as $\Delta^N$-sparse. 
We might then ask, under what conditions on what optimization problems can we recover $\xx_0$ and $U_N$ from $Y_N$, given the assumption that $U_N$ is $\Delta^N$-sparse? 

The optimization problem we focus on is the following, introduced in \cite{sefatiLinearSystemsSparse2015}:
\begin{equation}
    \min_{\xx_0, U_N} \|U_N\|_1\mbox{ s.t. } Y_N = \OO_N \xx_0 + \GGamma_N U_N 
    \tag{$D_1$}.
    \label{eqn:D1}
\end{equation}
This optimization problem may be thought of as an implementation of basis pursuit ($P_1$) for linear systems. Analogously, we would like to characterize the conditions under which this problem is well posed--i.e. no two entrywise $\Delta$-sparse inputs and generic initial conditions produce the same output--and when \eqref{eqn:D1} uniquely recovers such $\Delta$-sparse inputs and initial conditions jointly. The former condition may be thought of as an injectivity condition, in the sense that the output $Y_N$ uniquely specifies the initial condition and inputs up to $\Delta^N$-sparsity. Of additional interest are the cases in which the state space is not sufficiently observable to uniquely determine $x$ given $Y_N$, but we may still recover or uniquely characterize $\Delta^N$-sparse $U_N$. To these ends, we define the following:
\begin{definition}\label{def:system_conds}
Let $\Sigma = (\AA, \PPsi, \CC)$ be a linear system.
    \begin{itemize}
        \item $\Sigma$ is\textit{ jointly $\Delta^N$-injective} if $U,U' \in \mathcal{S}(\Delta^N)$ and $\OO_N\xx'_0 + \GGamma_NU'_N = \OO_N\xx_0 + \GGamma_NU_N \implies (\xx_0, U_N) = (\xx_0', U_N')$. We write the set of all such $\Sigma$ as $\mathcal{R}^*_0(\Delta, N)$.
        \item $\Sigma$ is \textit{input $\Delta^N$-injective} if $U,U' \in \mathcal{S}(\Delta^N)$ and $\OO_N\xx'_0 + \GGamma_NU'_N = \OO_N\xx_0 + \GGamma_NU_N \implies U_N = U_N'$. We write the set of all such $\Sigma$ as $\mathcal{R}_0(\Delta, N)$.
        \item $\Sigma$ is \textit{jointly $\Delta^N$-recoverable with \eqref{eqn:D1}} if any solution $(\xx_0, U_N)$ to \eqref{eqn:D1} s.t. $U_N \in \mathcal{S}(\Delta^N)$ is necessarily the unique solution. We write the set of all such $\Sigma$ as $\mathcal{R}^*_1(\Delta, N)$.
        \item $\Sigma$ is \textit{input $\Delta^N$-recoverable with \eqref{eqn:D1}} if any two solutions $(\xx_0, U_N), (\xx'_0, U'_N)$ to \eqref{eqn:D1} s.t. $U_N, U_N' \in \mathcal{S}(\Delta^N)$ satisfy $U_N = U_N'$. We write the set of all such $\Sigma$ as $\mathcal{R}_1(\Delta, N)$.
    \end{itemize}
\end{definition}

The condition established in \cite{sefatiLinearSystemsSparse2015} for when $\Sigma$ is jointly $\Delta_s^N$-recoverable with \eqref{eqn:D1} is based on the coherence $\mu: \mathbb{R}^{n\times m} \to [0,1]$, defined as
\begin{equation}
    \mu(\TTheta) = \max_{i \ne j} \frac{|\langle \TTheta_i, \TTheta_j \rangle|}{\|\TTheta_i\|_2\|\TTheta_j\|_2}
\end{equation}
Henceforth, we define $\PP_N^\perp := \II - \OO_N\OO_N^+$, the orthogonal projection onto the orthogonal complement of the column space of $\OO_N$. The main result of \cite{sefatiLinearSystemsSparse2015}, which may be implicitly read as a sufficient condition for $\PP_N^\perp \GGamma_N$ to satisfy the $Ns$-NUP, is as follows:
\begin{proposition}\label{prop:sefati}
    If $\ker{\OO_N} =0$ and $$\mu(\PP_N^\perp \GGamma_N)~<~\frac{1}{2Ns - 1}$$ then $\Sigma$ is jointly $\Delta_s^N$-recoverable with \eqref{eqn:D1}.
\end{proposition}

As is typical of coherence-based sparse recovery guarantees, it was found that this bound was enormously conservative for even modest $N$. It is also remarked that the condition in proposition \ref{prop:sefati} also guarantees that any $\xx_0$ and $U_N$ such that $|\supp(U_N)| \le Ns$ is also recovered uniquely by \eqref{eqn:D1}, so there is reason to suspect this condition could be tightened.

\section{NECESSARY AND SUFFICIENT CONDITIONS FOR JOINT STATE AND SPARSE INPUT RECOVERY}

Let $\Delta$ be an ASC over $\range{m}$ and $\Sigma$ a linear system as in \eqref{eq:system}. Our first observation is that the difference between joint and input-only recoverability/injectivity is just a matter of observability: 
\begin{lemma}\label{lemma:obs_equiv}
    $\Sigma \in \mathcal{R}_p^*(\Delta, N)$ if and only if $\ker \OO_N = 0$ and $\Sigma \in \mathcal{R}_p(\Delta, N)$.
\end{lemma}
\begin{proof}
    By definition, $\mathcal{R}_p^*(\Delta, N) \subset \mathcal{R}_p(\Delta, N)$. Suppose $\ker\OO_N \ne 0$, then $\xx_0$ cannot be uniquely determined by a constraint on $\OO_N\xx_0$, a contradiction. Hence $\Sigma \in \mathcal{R}_p^*(\Delta, N)$ implies $\ker \OO_N = 0$ and $\Sigma \in \mathcal{R}_p(\Delta, N)$.
    
    Now suppose $\Sigma \in \mathcal{R}_p(\Delta, N)$ and that $\ker\OO_N=0$. Then $\forall (\xx_0, U_N), (\xx_0', U_N')$ such that $U_N, U_N' \in \mathcal{S}(\Delta^N)$, $\OO_N\xx'_0 + \GGamma_NU'_N = \OO_N\xx_0 + \GGamma_NU_N \implies U_N = U_N'$, and therefore $\OO_N\xx_0 = \OO_N\xx_0'$. Since $\ker \OO_N = 0$, $\xx_0 = \xx_0'$.
\end{proof}

In \cite{sefatiLinearSystemsSparse2015}, implicit in the use of coherence for the main result was that when $\ker \OO_N = 0$, $\PP_N^\perp \GGamma_N$ satisfying the $Ns$-NUP is sufficient $\Sigma$ to be jointly $\Delta^N$-recoverable with \eqref{eqn:D1}. Per lemma \ref{lemma:obs_equiv}, we may suspect that if we relax the condition of observability, we may still determine a condition on when \eqref{eqn:D1} recovers the input, through conditions on $\PP_N^\perp \GGamma_N$. This is reflected in the fact that the following are equivalent:
\begin{gather}
    \exists \xx_0, \xx_0', \OO_N\xx_0 + \GGamma_N U_N = \OO_N\xx_0' + \GGamma_N U_N' \\
    \PP_N^\perp \GGamma_N U_N = \PP_N^\perp \GGamma_N U_N'
\end{gather}

For determining sparse input recoverability and injectivity then, we will see it is only necessary to consider properties of $\ker \PP_N^\perp \GGamma_N$.

In the case of regular sparsity, there are several equivalent ways to establish uniqueness of sparse solutions. The typical way is a rank condition on all collections of $2s$ columns of a matrix. Here we generalize this slightly, to $\Delta$-sparsity:

\begin{lemma}\label{lemma:delta_inj}
    Let $\Delta$ be an ASC over $\range{m}$.
    $\forall \xx, \xx' \in \mathcal{S}(\Delta), \TTheta\xx = \TTheta\xx' \implies \xx = \xx'$ if and only if $\forall S, S' \in \Delta$, $\ker \TTheta_{S \cup S'} = 0$. When either condition is satisfied, we say $\TTheta$ is $\Delta$-injective.
\end{lemma}
\begin{proof}
Suppose $\exists S, S' \in \Delta$, $\ker \TTheta_{S \cup S'} \ne 0$, then let $\xx, \xx' \in \RR^m$ distinct such that $\supp\xx = S$ and $\supp\xx' = S'$, and $\TTheta(\xx - \xx') = 0$. Then $\TTheta\xx = \TTheta\xx'$ but $\xx\ne\xx'$, a contradiction. Conversely, suppose $\xx, \xx' \in \mathcal{S}(\Delta)$ such that $\TTheta\xx = \TTheta\xx'$ but $\xx \ne \xx'$, then where $S = \supp\xx$ and $S'=\supp\xx'$, $\TTheta_{S\cup S'}(\xx - \xx')_{S\cup S'} = 0$, so $\ker\TTheta_{S\cup S'} \ne 0$.
\end{proof}
This notion enables concise necessary and sufficient characterizations of when $\Delta$-sparse solutions are unique, and we will put it to good use in Theorem \ref{thm:l0}.

Shortly after the publication of \cite{sefatiLinearSystemsSparse2015}, Kafashan et al. produced sufficient conditions for both joint $\Delta_s^N$-recovery with \eqref{eqn:D1} and a related optimization problem incorporating the assumption of noise \cite{kafashanRelatingObservabilityCompressed2016}. Lemma 1 of \cite{kafashanRelatingObservabilityCompressed2016} gets close to being a necessary condition when $\Delta=\Delta_s$, but it makes use of the restricted isometry constant rather than purely rank-type constraints, and is not stated as necessary and sufficient. We provide a version of the necessary and sufficient condition in the style of \cite{kafashanRelatingObservabilityCompressed2016} to reflect this contribution, making only the change of the restricted isometry condition to $\Delta$-injectivity:
\begin{theorem}[$\Delta^N$-Injectivity of $\Sigma$]\label{thm:l0}
    Let $\Delta$ be an ASC over $\range{m}$ and $\Sigma = (\AA, \PPsi, \CC)$ a linear system with state space of dimension $n$. The following are equivalent:
    \begin{enumerate}
        \item $\Sigma$ is jointly $\Delta^N$-injective
        \item $\forall S, S' \in \Delta^{N}$, $\rank{[\OO_N \,\, (\GGamma_N)_{S\cup S'}]} = n + \rank((\GGamma_N)_{S \cup S'})$ and $C\PPsi$ is $\Delta$-injective
        \item $\ker \OO_N = 0$ and $\PP_N^\perp \GGamma_N$ is $\Delta^N$-injective.
    \end{enumerate}
\end{theorem}
\begin{proof}
    Suppose $\Sigma$ is jointly $\Delta^N$-injective, then if $U_N, U_N' \in \mathcal{S}(\Delta)$, $\OO_N\xx_0 + \GGamma_N U_N = \OO_N\xx_0' + \GGamma_N U_N' \implies \xx_0' = \xx_0$ and $U_N = U_N'$. It follows that for every $S, S' \in \Delta^N$, $[\OO_N \,\, (\GGamma_N)_{S\cup S'}]$ is full column rank, so $\rank{[\OO_N \,\, (\GGamma_N)_{S\cup S'}]} = n + \rank((\GGamma_N)_{S \cup S'})$, and $\ker(\GGamma_N)_{S \cup S'} = 0$. It follows that every $((\GGamma_N)_k)_{S_k \cup S_k'}$ is full rank, so $(\OO_{N-k}\PPsi)_{S_k \cup S_k'}$ is full rank $\forall k \in \range{N}$. In particular we may take $k=0$, so $\forall S, S' \in \Delta$, $(\OO_{0}\PPsi)_{S \cup S'} = (\CC\PPsi)_{S \cup S'}$ is full rank, and therefore $\CC\PPsi$ is $\Delta$-injective. 
    
    Likewise, if $(\CC\PPsi)_{S \cup S'}$ is full rank, since $\GGamma_N$ is block triangular with $\CC\PPsi$ on the diagonal, we have that $(\GGamma_N)_{S \cup S'}$ is full rank for every $S, S' \in \Delta^N$. It follows that if $\rank{[\OO_N \,\, (\GGamma_N)_{S\cup S'}]} = n + \rank((\GGamma_N)_{S \cup S'})$, then $[\OO_N \,\, (\GGamma_N)_{S\cup S'}]$ is full rank, so if $U_N$ is $\Delta^N$-sparse and $\xx_0$ is generic such that $Y_N = \OO_N\xx_0 + \GGamma_N U_N$, they are unique. Therefore, $\Sigma$ is jointly $\Delta^N$-injective. We have thus shown that $(1) \iff (2)$. 
    
    Note that from the above, we have $(2)$ if and only if $\forall S, S' \in \Delta^{N}$, $\ker{[\OO_N \,\, (\GGamma_N)_{S\cup S'}]} = 0$, which is the case if and only if $\ker \OO_N = 0$, $\ker (\GGamma_N)_{S\cup S'}=0$, and $\img \OO_N \cap \img (\GGamma_N)_{S\cup S'} = 0$. The third point holds iff $\rank \PP_N^\perp (\GGamma_N)_{S\cup S'} = \rank (\GGamma_N)_{S\cup S'}$, so this is again equivalent to $\ker \OO_N = 0$ and $\ker \PP_N^\perp (\GGamma_N)_{S\cup S'}=\ker (\PP_N^\perp\GGamma_N)_{S\cup S'}=0$. We conclude by noting that this is equivalent to $\ker \OO_N = 0$ and $\PP_N^\perp \GGamma_N$ being $\Delta^N$-injective.

\end{proof}

Having shown the result for uniqueness of solutions/sparse injectivity, we proceed to the problem of recoverability. Recalling that we are only interested in $\Delta^N$-sparse solutions $U_N$, we could obtain a necessary and sufficient condition from the GNUP if this support pattern is an ASC. Technically, as $\Delta^N$ is a set of tuples of sets, it cannot be a simplicial complex, but if one instead interprets these tuples as disjoint unions, we uncover an ASC structure:

\begin{lemma}
     $\Delta^N$ is a simplicial complex up to bijection.
\end{lemma}
\begin{proof}
    It is clear that every $(S_k)_{k \in \range{N}} \in \Delta^N$ can be identified with the set $\tilde{S} = \bigcup_{k \in \range{N}} \{k\} \times S_k$. Let $S\in \Delta^N$ and suppose $\tilde{S}' \subset \tilde{S}$. Then for all $k$, $\{k\} \times S_k' \subset \{k\} \times S_k \implies S_k' \subset S_k$, so $S_k' \in \Delta$. Therefore $S' \in \Delta^N$.
\end{proof}
It is then clear that we can apply the GNUP in this context, to obtain necessary and sufficient conditions on $\Delta^N$-recovery.

\begin{theorem}[$\Delta^N$-Recoverability with \eqref{eqn:D1}]\label{thm:l1}
    Let $\Delta$ be an ASC over $\range{m}$ and $\Sigma = (\AA, \PPsi, \CC)$ a linear system. The following are equivalent:
    \begin{enumerate}
        \item $\Sigma$ is jointly $\Delta^N$-recoverable with \eqref{eqn:D1}.
        \item $\ker \OO_N = 0$ and $\PP_N^\perp \GGamma_N$ satisfies the $\Delta^{N}$-NUP.
    \end{enumerate}
\end{theorem}
\begin{proof}
    Recall that $\Sigma$ is jointly $\Delta^N$-recoverable if and only if $\ker \OO_N = 0$ and for any two solutions $(\xx_0, U_N), (\xx_0, U_N')$ to $D_1(\Sigma, N)$ such that $U_N, U_N' \in \mathcal{S}(\Delta^N)$ satisfy $U_N = U_N'$. This is equivalent to any $\Delta^N$-sparse solution to the optimization problem $\min_{U_N}\|U_N\|_1$ s.t. $\exists \xx_0, Y_N = \OO_N\xx_0 + \GGamma_N U_N$ being the unique solution, and we have that $\exists \xx_0, Y_N = \OO_N\xx_0 + \GGamma_N U_N$ if and only if $\PP_N^\perp Y_N = \PP_N^\perp \GGamma_N U_N$. Therefore, any $\Delta^N$-sparse solution to $\min_{U_N}\|U_N\|_1$ s.t. $\exists \xx_0, Y_N = \OO_N\xx_0 + \GGamma_N U_N$ is necessarily unique if and only if $\PP_N^\perp \GGamma_N$ satisfies the $\Delta^N$-NUP, and we may conclude.
\end{proof}

As system-theoretic statements, these conditions establish a deadbeat unknown-input state estimator and input reconstructor for linear systems with sparse inputs, analogous to the generic input version in \cite{ansariDeadbeatInputReconstruction2019}. From the point of view of sparse recovery, these conditions mirror the role of the spark and the standard nullspace property for $P_0$, $P_1$. Like these conditions, they are clearly NP-hard to verify; however, it is clear that in the case of regular sparsity, the number of supports to check to verify the $\Delta_s^N$-NUP is far smaller than for the $Ns$-NUP. We will see this idea reflected in the next section: in many cases, it indeed suffices to check much easier conditions. 

\section{INTERPRETABLE CONDITIONS FOR $\Delta^N$-RECOVERABILITY WITH \eqref{eqn:D1}}

In the last section, we saw that by casting $\Sigma$ as a block matrix-vector system, standard sparse recovery arguments yield a spark-like condition for $\Delta^N$-injectivity, and subsequent application of the generalized nullspace property yields an analogous condition for $\Delta^N$-recoverability. Once we have $\ker \OO_N = 0$, each of these conditions are entirely determined by attributes of $\ker P_N^\perp \GGamma_N$. This space is precisely the subspace of inputs $U_N$ for which there exists an initial condition $\xx_0$, $\GGamma_N U_N = \OO_N \xx_0$. By constructing lower and upper bounds on $\nsc (P_N^\perp \GGamma_N, \Delta^N)$ via $\nsc(\CC\PPsi, \Delta)$ and $\nsc ((\CC\PPsi)^{-1}\CC\AA\ker\CC, \Delta)$, we obtain conditions that are respectively necessary and sufficient for $\Sigma$ to be input $\Delta^N$-recoverable.

\subsection{A Necessary Condition}
In \cite{sefatiLinearSystemsSparse2015}, it was empirically observed that $\mu(\CC\PPsi) \le \mu(\GGamma_N) \le \mu(\PP_N^\perp \GGamma_N)$. The intuition one is tempted to derive from this is that a more incoherent $\CC\PPsi$ indicates a greater chance of \eqref{eqn:D1} succeeding. To concretize this idea, one need look no further than the question of recovering the very last input to a system, left untouched by the system dynamics. Specifically, it is useful to consider the one-to-one correspondence between elements of $\ker \CC\PPsi$ and inputs $U_N \in \ker \PP_N^\perp \GGamma_N$ with all zero entries except for the last:

\begin{lemma}\label{lemma:cpsi}
    Suppose $\uu$ is an input such that $\forall k < N, \uu_k = 0$ and $\uu_{N-1} \in \ker \CC\PPsi$. Then $U_N \in \ker \PP_N^\perp \GGamma_N$.
\end{lemma}
\begin{proof}
    Suppose $U_N$ is as described, then with $\xx_0 = 0$, $k < N \implies \xx_k = 0 \implies \yy_k = 0$. We also have that $\yy_{N} = \CC \xx_{N} = \CC \PPsi \uu_{N-1} = 0$. Therefore, $U_N \in \ker \GGamma_N \subset \ker  \PP_N^\perp \GGamma_N$.
\end{proof}
This suggests the following natural necessary condition:
\begin{proposition}\label{prop:cpsi_necc}
    $\Sigma$ can only be input $\Delta^N$-recoverable with \eqref{eqn:D1} (resp. injective) if $\CC\PPsi$ satisfies the $\Delta$-NUP (resp. is $\Delta$-injective)
\end{proposition}
\begin{proof}
    The case of injectivity is implied by characterization 2 of Theorem \ref{thm:l0}.
    Suppose $\Sigma$ is $\Delta^N$-recoverable. Then by lemma \ref{lemma:cpsi}, for any input $\uu$ such that $\forall k < N$, $\uu_k = 0$ and $\uu_{N-1} \in \ker\CC\PPsi$, $U_N$ is the unique input component of the solution to \eqref{eqn:D1} if $U_N$ is $\Delta^N$-sparse, so $\min_{U} \|U_N\|_1 = \min_{\vv} \|\vv\|_1$ s.t. $\CC\PPsi \uu_{N-1} = \CC\PPsi \vv$ always recovers $\Delta$-sparse $\uu_{N-1}$, hence $\CC\PPsi$ satisfies the $\Delta$-NUP.
\end{proof}
Another way to put this is in terms of the nullspace constant:
\begin{corollary}\label{corr:ness}
$\Sigma$ can only be input $\Delta^N$-recoverable if $\nsc(\CC\PPsi, \Delta) < 0.5$.
\end{corollary}
It is clear that if $\CC\PPsi$ does not satisfy $\Delta$-NUP, there will always exist entrywise $\Delta$-sparse inputs which cannot be recovered by the problem \eqref{eqn:D1}. As this will be the case for any $N$, this is in a sense the tightest necessary condition for this class of problems.

\subsection{A Sufficient Condition}
We will now show that whenever the output of a system is uniformly zero, at each time point $k$ \textit{regardless} of initial condition, one must have $\uu_k \in (\CC\PPsi)^{-1}\CC\AA\ker\CC$:

\begin{lemma}\label{lemma:suff}
    Let $\uu$ be an input such that $U_N \in \ker \PP_N^\perp \GGamma_N$. Then $\forall k < N$, $\uu_k \in (\CC\PPsi)^{-1} \CC\AA \ker \CC$.
\end{lemma}
\begin{proof}
    Let $\uu$ be as above, and let $\xx_0$ be such that $\OO_N \xx_0 + \GGamma_N U_N = Y_N = 0$. Suppose the claim is false. Then $\exists k \le N-2$, $\CC\PPsi \uu_k \notin \CC\AA \ker \CC$, so if $\xx_k \in \ker \CC$, $\CC(\AA \xx_k + \PPsi \uu_k) = \yy_{k+1} \ne 0$. Since we assumed $Y_N = 0 \implies \yy_{k+1} = 0$, this cannot be, so $\xx_k \notin \ker \CC$, which gives our contradiction as $\yy_k = \CC\xx_k \ne 0$. 
\end{proof}

A fact not immediately visible is that this subspace arises as the kernel of $\PP_1^\perp \GGamma_1$:

\begin{lemma}\label{lemma:equiv_proj}
$$\ker \PP_1^\perp \GGamma_1 = (\CC\PPsi)^{-1}\CC\AA\ker\CC$$
\end{lemma}
\begin{proof}
\begin{align*}
\ker \PP_1^\perp \GGamma_1 
&= \{\uu: \GGamma_1 \uu \in \img \OO_1\} \\
&= \{\uu: \exists \xx_0, \CC \xx_0 = 0\mbox{ and } \CC\PPsi\uu = \CC\AA\xx_0 \} \\
&= \{\uu: \CC\PPsi\uu \in \CC\AA\ker\CC \}\\
&= (\CC\PPsi)^{-1}\CC\AA\ker\CC
\end{align*}
\end{proof}
This leads us to the following:
\begin{proposition}\label{prop:gnup_alln}
    $\Sigma$ is input $\Delta^N$-recoverable for every $N$ if and only if $\PP_1^\perp \GGamma_1$ satisfies the $\Delta$-NUP.
\end{proposition}
\begin{proof}
    $\Sigma$ is input $\Delta^N$-recoverable for every $N$ only if $\Sigma$ is input $\Delta$-recoverable; by Theorem \ref{thm:l1}, we have the forward implication. Since lemma \ref{lemma:suff} ensures that every entry $\uu_k$ of an input $U_N \in \ker \PP_N^\perp \GGamma_N$ satisfies $\uu_k \in (\CC\PPsi)^{-1}\CC\AA\ker\CC \overset{\text{lemma \ref{lemma:equiv_proj}}}{=} \ker \PP_1^\perp \GGamma_1$, for any $S \in \Delta^N$, 
    \begin{align*}
        \|(U_N)_S\|_1 &= \sum_{k \in \range{N}}\|(\uu_k)_{S_k}\|_1 \\
        &< \sum_{k \in \range{N}}\nsc(\PP_1^\perp \GGamma_1, \Delta)\|\uu_k\|_1 \\
        &= \nsc(\PP_1^\perp \GGamma_1, \Delta)\|U_N\|_1 \\
        \implies \frac{\|(U_N)_S\|_1}{\|U_N\|_1} &< \nsc(\PP_1^\perp \GGamma_1, \Delta) < \frac{1}{2}\\
        \implies& \nsc(\PP_N^\perp \GGamma_N, \Delta^N) < \frac{1}{2}
    \end{align*}
    therefore $\PP_N^\perp \GGamma_N$ satisfies the $\Delta^N$-NUP, and we conclude by Theorem \ref{thm:l1}.
    
\end{proof}
This may be also be phrased as a sufficient condition for any $N$, in terms of the nullspace constant:
\begin{corollary}\label{corr:suff}
    $\Sigma$ is input $\Delta^N$ recoverable if $\nsc(\PP_1^\perp \GGamma_1, \Delta) < 0.5$.
\end{corollary}
Since this condition is clearly also necessary for $N=1$, we might expect this to be closer to necessity for smaller values of $N$. 

\subsection{Tight Cases}
While these conditions only form the tip of the iceberg in terms of possible guarantees for $\Delta^N$-recovery, they are unlike existing conditions in the literature as they each achieve necessity and sufficiency for certain classes of systems. Implicit in proposition \ref{prop:gnup_alln} is the necessary and sufficient condition for input $\Delta$-recoverability, though this can only indicate joint $\Delta$-recoverability when $\rank \OO_1 = n$. The major condition leading to these conditions being tight is when $\ker\CC\PPsi = (\CC\PPsi)^{-1}\CC\AA\ker\CC$, which is equivalent to $\AA\ker\CC \subset \ker\CC$--that is, the system is unobservable. This hints at an intuition that systems that are more observable, e.g. systems that quickly divulge information about their initial conditions, may actually result in nullspace constants of these spaces that are further apart. We also note that $(\CC\PPsi)^{-1}\CC\AA\ker\CC = (\CC\PPsi)^+\CC\AA\ker\CC + \ker \CC\PPsi$, so characteristics of the space $(\CC\PPsi)^+\CC\AA\ker\CC$ directly mediate how close the conditions established in this section are to one another. We might then expect that the gap between them to scale with $\dim\ker\CC$, and therefore to be small when $p$ is not much less than $n$.

\section{NUMERICAL VALIDATION}
To validate and provide intuition for our results, we perform two types of numerical experiments on ensembles of random LDSs with $s$-sparse inputs, i.e. $\Delta$-sparse with $\Delta \in \{\Delta_s: s \in \mathbb{N}\}$:
\begin{enumerate}
    \item We evaluate the ability of \eqref{eqn:D1} to jointly recover the initial state and sparse inputs for different system parameters and sparsity levels.
    \item We explore the relationship between $\nsc(\CC\PPsi, \Delta_s)$ and $\nsc(\PP_1^\perp \GGamma_1, \Delta_s)$ as a function of system parameters and sparsity levels.
\end{enumerate}
Throughout, we employ the 1-step TSA branch-and-bound algorithm of \cite{choComputablePerformanceGuarantees2018} to compute nullspace constants to a tolerance of $\pm 0.05$.  This level of precision was chosen to permit the estimation of nullspace constants for much larger sparsity levels with a statistically significant amount of systems than would be permitted if we were to compute them exactly. 

\subsection{System Generation and Computing Information}
For a given $n, m, p$, systems $\Sigma = (\AA, \PPsi, \CC)$ were generated randomly, mirroring the strategy of \cite{sefatiLinearSystemsSparse2015}:
\begin{itemize}
    \item $\AA$ is i.i.d. Gaussian in each entry with variance $1/n$, s.t. each eigenvalue $\lambda$ of $\AA$ satisfies $|\lambda| < 0.9$.
    \item $\PPsi$ is i.i.d Gaussian in each entry with variance $1/n$.
    \item $\CC$ is i.i.d. Gaussian with unit variance.
\end{itemize}
For all systems considered, we choose $m=20$, to ensure tractability of $\Delta_s$-nullspace constant estimates. We also restrict our attention to $n \le m$ and $p \le m$, focusing on the usual assumption of ``overcompleteness'' of $\PPsi$. We analyze a total of 25173 systems across all parameter ranges. We conducted all experiments and analysis in Python, using CVXPY with the GUROBI solver for implementation of linear programming.

\subsection{Success of recovery}
The first analysis we present provides numerical support for the claims of necessity and sufficiency of the conditions established in corollaries \ref{corr:ness} and \ref{corr:suff}, that $\nsc(\CC\PPsi,\Delta_s) < 0.5$ is necessary and $\nsc(\PP_1^\perp \GGamma_1,\Delta_s) < 0.5$ is sufficient for $\Sigma$ to be input $\Delta_s^N$-recoverable.

\begin{table}
    \centering  
    \vspace{2mm}
    \caption{Number of systems with at least one imperfect joint recovery across 30 trials out of total occurrences of systems, with given ranges of $\sigma_{\CC\PPsi}:= \nsc(\CC\PPsi, \Delta_s) $ and $\sigma_{\PP_1^\perp \GGamma_1} := \nsc(\PP_1^\perp\GGamma_1, \Delta_s)$. $n \in \{12, 16, 20\}$, $1 \le p \le 20$, $1 \le s \le 5$ and $m=20$. $\sigma \sim 0.5$ if $|\sigma - 0.5| \le 0.05$ (experimental tolerance). }
    
    \begin{tabular}{c|c|c|c}
         & $\sigma_{\CC\PPsi} < 0.5$ & $\sigma_{\CC\PPsi} \sim 0.5$ & $\sigma_{\CC\PPsi}  > 0.5$ \\
         \hline
         $\sigma_{\PP_1^\perp \GGamma_1} > 0.5$ & $153/800$ & $34/289$ & $9068/9866$ \\
         \hline
         $\sigma_{\PP_1^\perp \GGamma_1}\sim 0.5$ & $0/116$ & $60/336$ & $2/10$ \\
         \hline
         $\sigma_{\PP_1^\perp \GGamma_1} < 0.5$ & $0/2100$ & $0/0$ & $0/0$
    \end{tabular}

    \label{fig:recovery_stats}
\end{table}

For each system, we choose a random $s \le 10$, and simulate 30 random combinations of initial conditions and entrywise $s$-sparse inputs of length $N=n$. As in \cite{sefatiLinearSystemsSparse2015}, we sample each entry of the initial conditions uniformly from $[-5,5]$, choose every support $S_k \subset \range{m}$ uniformly without replacement s.t. $|S_k|=s$, and sample each nonzero input entry uniformly on $[-5, 5]$. In Table \ref{fig:recovery_stats}, we provide the fraction of systems for which there was at least one imperfect joint recovery of the 30 trials conducted as a function of the nullspace constants $\nsc(\CC\PPsi, \Delta_s)$ and $\nsc(\PP_1^\perp \GGamma_1, \Delta_s)$, which were determined up to being less than $0.5$, greater than $0.5$, or in $[0.45, 0.55]$. From the table, we see our conditions behave as expected: all systems with $\nsc(\PP_1^\perp \GGamma_1, \Delta_s) < 0.5$ exhibit perfect recovery; a majority of systems with $\nsc(\PP_1^\perp \GGamma_1, \Delta_s) > 0.5$ exhibit imperfect recovery; and $\nsc(\PP_1^\perp \GGamma_1, \Delta_s) \ge \nsc(\PP_1^\perp \GGamma_1, \Delta_s)$ for all systems. In particular, there were many systems such that, for the sparsity level considered, $\nsc(\CC\PPsi, \Delta_s) < 0.5$ was not sufficient for perfect input recovery, as $\nsc(\PP_1^\perp \GGamma_1, \Delta_s) > 0.5$.

\begin{figure}[t]
    \centering
    \vspace{3mm}
    \includegraphics[width=0.5\textwidth]{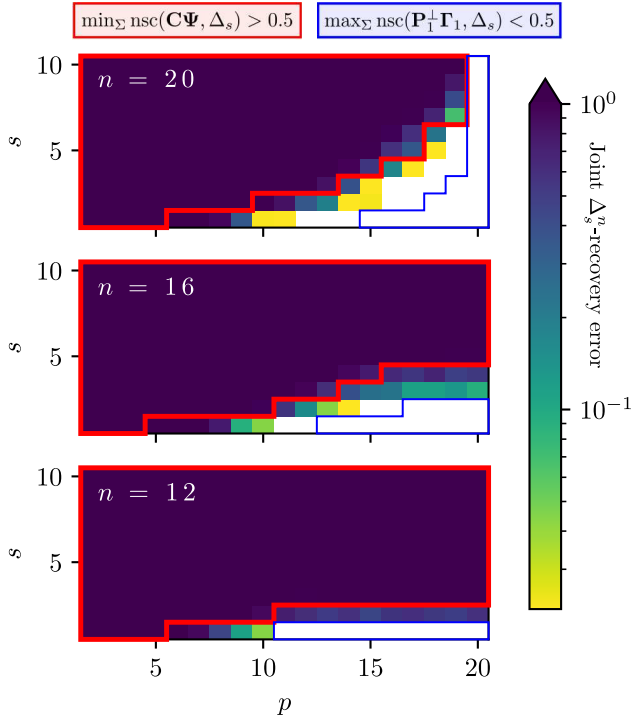}
    \caption{Empirical probability of imperfect joint recovery of entrywise $\Delta_s$-sparse signals and generic initial conditions as a function of $s$ and $p$ with $m~=~20$, $n \in \{12, 16, 20\}$, and $N=n$; white indicates all signals perfectly recovered. Within the red outlines are $(s,p)$ such that all simulated systems satisfied $\nsc(\CC\PPsi, \Delta_s) > 0.5$. Within the blue outline are $(s,p)$ such that all simulated systems satisfied $\nsc(\PP_1^\perp\GGamma_1, \Delta_s) < 0.5$.}
    \label{fig:phaseplot}
\end{figure}

Figure \ref{fig:phaseplot} visualizes this data in a different way, presenting a joint recovery phase transition plot over $p$ and $s$ for $n \in \{12, 16, 20\}$ and $N=n$. The intensity of each pixel indicates the empirical probability of a system with dimensions $p, n, m$ and sparsity level $s$ exhibiting imperfect joint recovery. The colormap normalization is in log-scale to better illustrate the case of perfect recovery for all trials, which we plot as white. The red and blue outlined regions indicate $(p,n,m,s)$ such that every system simulated satisfied $\nsc(\CC\PPsi, \Delta_s) > 0.5$ (red, failure of necessary condition) or $\nsc(\PP_1^\perp\GGamma_1, \Delta_s) < 0.5$ (blue, satisfaction of sufficient condition). We observe that there are no white pixels contained in the red outlined regions, supporting the claim of $\nsc(\CC\PPsi, \Delta_s) > 0.5$ being necessary for joint $\Delta_s^n$-recoverability. The red outlined region also becomes strictly smaller as $n$ increases, as expected as this indicates $\CC\PPsi$ with generically higher rank. Likewise, we observe that there are no colored pixels contained in the blue outlined regions, supporting the claim that $\nsc(\PP_1^\perp\GGamma_1, \Delta_s) < 0.5$ is sufficient for joint $\Delta_s^n$-recoverability. The blue regions appear to shrink to the right as $n$ increases, widening the gap in $s$ between the red and blue outlined regions for smaller $p$. This reflects the fact that increased $n$ will generically result in an increase in the dimension of $\ker\CC$, and thus an increase in the dimension of $\ker \PP_1^\perp \GGamma_1$.

\subsection{Analysis of NSC Relationship}
We conclude our analysis by taking a closer look at the relationship between the two main quantities motivated in this work, $\nsc(\CC\PPsi, \Delta_s)$ and $\nsc(\PP_1^\perp \GGamma_1, \Delta_s)$. Here we do not terminate once determining whether the constant is above or below 0.5, as was done for the previous section, opting instead to compute bounds on these constants for each system up to a tolerance of $\pm 0.05$.
In figure \ref{fig:nsc_compare}, for various $n,m,p, s$ we plot $\nsc(\CC\PPsi, \Delta_s)$ against $\nsc(\PP_1^\perp \GGamma_1, \Delta_s)$ . Each plotted point represents the midpoint of the computed bounds on these constants for a given $\Sigma$ and sparsity level $s$. We see that for every set of parameters considered, the nullspace constants tend to increase with the sparsity level, and that $\nsc(\CC\PPsi, \Delta_s) \le \nsc(\PP_1^\perp \GGamma_1, \Delta_s)$ as expected from the fact that $\ker \CC\PPsi \subset (\CC\PPsi)^+\CC\AA\ker\CC + \ker \CC\PPsi = \ker \PP_1^\perp \GGamma_1$. 

In the top set of figures we fix $n=19$ and $m=20$, and illustrate the trend of $\nsc(\CC\PPsi, \Delta_s)$ and $\nsc(\PP_1^\perp \GGamma_1, \Delta_s)$ as $p$ increases. For $p=11$, $\nsc(\PP_1^\perp \GGamma_1, \Delta_s)$ is quite pessimistic, echoing the wide gap between the red and blue regions of figure \ref{fig:phaseplot} for $n=20$ for moderate $p$. As $p$ increases, points approach the diagonal and tend to decrease in magnitude, reflecting the intuitively more favorable recovery properties of a larger number of measurements. When $p=n$, we see that the two constants become equal; this reflects the fact that $\CC \in \mathbb{R}^{p \times n}$ will be generically full rank and thus $\ker \PP_1^\perp \GGamma_1 = (\CC\PPsi)^+\CC\AA\ker\CC + \ker \CC\PPsi = \ker \CC\PPsi$.

In the bottom set of figures we fix $p=11$ and $m=20$, and illustrate the trend of $\nsc(\CC\PPsi, \Delta_s)$ and $\nsc(\PP_1^\perp \GGamma_1, \Delta_s)$ as $n$ increases. Up to our tolerance, we see again that for $n=11$, $n=p$ and so the constants are equal. With the increase in $n$, we observe that $\nsc(\CC\PPsi, \Delta_s)$ stays roughly fixed, while $\nsc(\PP_1^\perp \GGamma_1, \Delta_s)$ increases. This is qualitatively similar to the effect of decreasing $p$, but in that case, \textit{both} constants were affected. This comes as no surprise, as for $p \le n \le m$ we expect the distribution of $\CC\PPsi$ to be relatively unaffected, but $\dim\ker\CC$ to increase, and thus for $\dim \ker \PP_1^\perp \GGamma_1$ to increase. Another, more intuitive interpretation is that by increasing $n$ with $p$ fixed, there is an increased amount of interference from the initial condition mixed in with the sparse inputs, and so guarantees on the recovery of even the sparse inputs alone will naturally worsen for recovery with a given number of time steps $N$. The fact that recovery empirically does better across the board for larger $n$ in figure \ref{fig:phaseplot} is explained by the fact that with a larger state space, a larger number of time steps will continue to yield insight about previous inputs and the initial condition.

Overall, these plots support the point that propositions \ref{prop:cpsi_necc} and \ref{prop:gnup_alln} can actually be fairly tight, provided that $p$ is not much less than $n$. Larger $p$ is in general seen to be better across the board, while increasing $n$ relative to $p$ results in less favorable sufficient guarantees. Importantly, they also reinforce the fact that in most, but not all, cases, both $\nsc(\CC\PPsi, \Delta_s)$ and $\nsc(\PP_1^\perp \GGamma_1, \Delta_s)$ will be either less than or greater than $0.5$. The cases for which they disagree are few, indicating potential practical utility for these quantities for determining joint $\Delta^N$-recovery with \eqref{eqn:D1} in applications.

\begin{figure*}[!ht]
  \centering
  \begin{minipage}{0.35\textwidth}
  \begin{subfigure}{\textwidth}
    \centering
    \includegraphics[width=\linewidth]{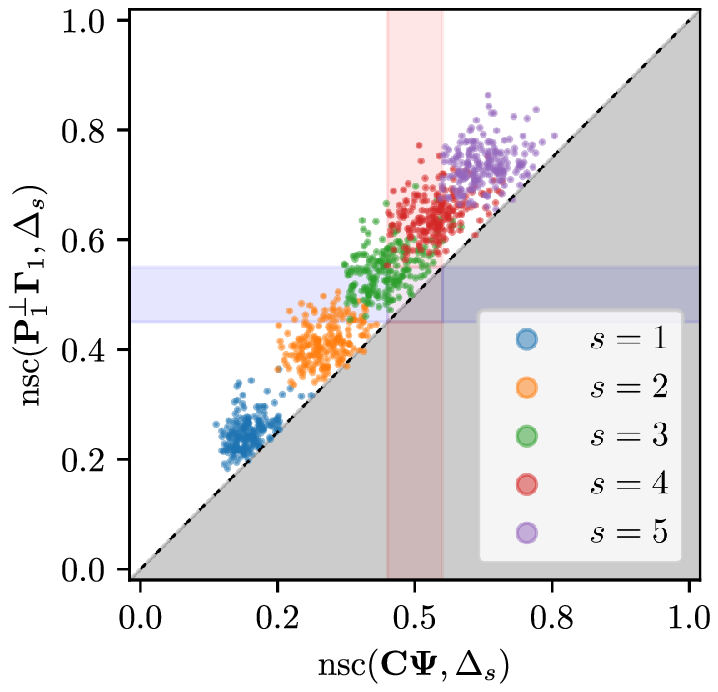}
    \caption{Example: $n=19, m=20, p=17$}
    \label{fig:example}
  \end{subfigure}%
    \end{minipage}
  \hfill
\begin{minipage}{.64\linewidth}
    \begin{subfigure}{\textwidth}
    \includegraphics[width=\linewidth]{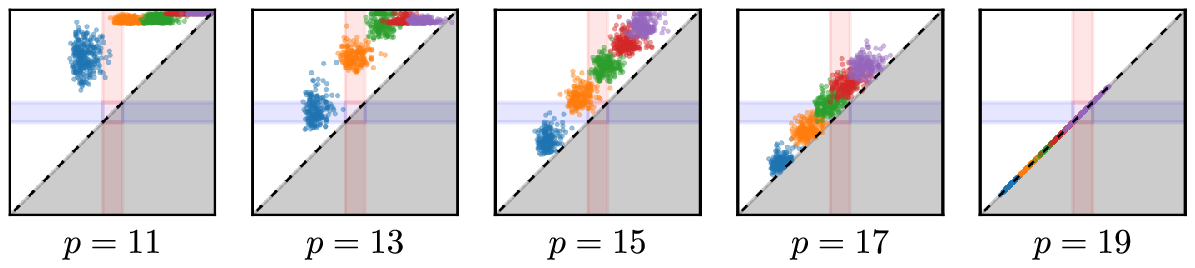}
    \caption{Trend as $p$ increases ($n=19, m=20$)}
    \label{fig:p_inc}
  \end{subfigure}\\
  \begin{subfigure}{\linewidth}
    \includegraphics[width=\linewidth]{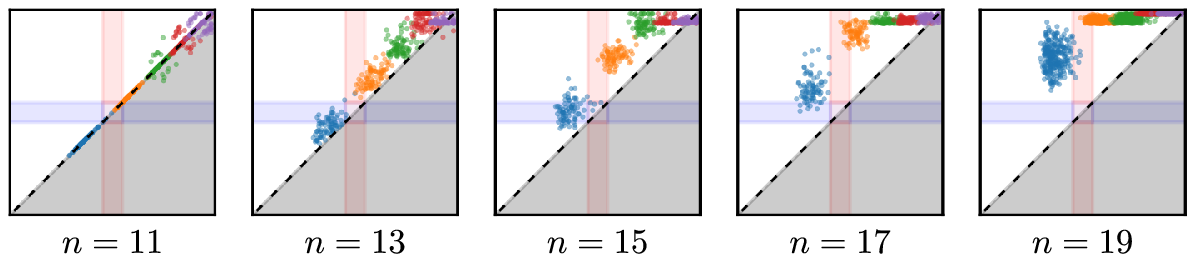}
    \caption{Trend as $n$ increases ($p=11$, $m=20$)}
    \label{fig:n_inc}
  \end{subfigure}
  \end{minipage}

    \caption{Distribution of $\Delta_s$-nullspace constants for random systems up to a tolerance of $0.05$ (11656 total systems). (\ref{fig:example}) shows in detail the case of $n=19, m=20, p=17$. The red region is where $|\nsc(\CC\PPsi, \Delta_s) - 0.5| < 0.05$, and the blue region indicates $\nsc(\PP_1^\perp \GGamma_1, \Delta_s) - 0.5| < 0.05$. White regions indicate that $\nsc(\CC\PPsi, \Delta_s), \nsc(\PP_1^\perp\GGamma_1, \Delta_s)$ are each greater/less than 0.5 with certainty. The shaded region under the diagonal corresponds to the impossible case $\nsc(\CC\PPsi, \Delta_s) > \nsc(\PP_1^\perp\GGamma_1, \Delta_s)$. 
    (\ref{fig:p_inc}) illustrates that as $p$ increases, both $\nsc(\CC\PPsi, \Delta_s), \nsc(\PP_1^\perp\GGamma_1, \Delta_s)$ tend to decrease, with $\nsc(\PP_1^\perp\GGamma_1, \Delta_s)$ converging to $\nsc(\CC\PPsi, \Delta_s)$. 
    (\ref{fig:n_inc}) illustrates that as $n$ increases, $\nsc(\PP_1^\perp\GGamma_1, \Delta_s)$ increases while $\nsc(\CC\PPsi, \Delta_s)$ is relatively unaffected.}
  \label{fig:nsc_compare}
\end{figure*}

\section{CONCLUSIONS AND FUTURE WORK}

We have presented the first necessary and sufficient conditions
under which the joint recovery of generic initial conditions and sparse inputs for a linear dynamical system is well-posed and may be carried out via $\ell_1$-minimization. Leveraging this characterization, we further provide two simple necessary, and sufficient conditions for joint $\Delta^N$-recoverability.
In contrast to previous work, these conditions have intuitive justifications and can be computationally verified for most systems.
Through $\ell_1$-recovery experiments on random systems, we showed that these conditions are useful indicators of recovery performance. 

We have identified several exciting avenues for future research building on the necessary and sufficient conditions introduced here. One direction is to extend these results to general systems with feedthrough terms (nonzero $\mathbf{D}$ matrix) and broaden the scope of recovery results to incorporate recovery with delay. More generally, we aim to explore notions of strong observability for linear systems with sparse inputs. Indeed, there are many significant results in linear systems theory that we believe admit direct extensions to the case of sparse inputs, and for which necessary and sufficient conditions using $\ell_1$-minimization and related characterizations may be achieved by building on the generalized nullspace property, just as was illustrated here.

\section*{ACKNOWLEDGEMENTS}
The authors would like to thank Dr. Adam Charles for invaluable discussions during the preparation of this manuscript.

\bibliographystyle{IEEEtran}
\bibliography{bibliography-final}

\end{document}